\documentclass[reqno]{amsart}
\usepackage[foot]{amsaddr}
\usepackage{graphicx}
\graphicspath{ {./figures/} }
\usepackage[margin=3cm]{geometry}
\usepackage{amsmath, amssymb,amsthm}
\usepackage[shortlabels]{enumitem}
\usepackage{mlmodern}
\DeclareSymbolFont{largesymbols}{OMX}{cmex}{m}{n} 
\usepackage{mathtools} 
\usepackage{stmaryrd} 
\usepackage{tikz}
\usetikzlibrary{decorations.pathreplacing,arrows.meta,external}
\usepackage{upref} 
\usepackage[colorlinks]{hyperref}
\hypersetup{citecolor=blue,filecolor=blue,linkcolor=blue,urlcolor=navyblue}
\definecolor{navyblue}{rgb}{0.0, 0.0, 0.5}

\newtheorem{thm}{Theorem}[section]
\newtheorem{prop}[thm]{Proposition}
\newtheorem{lem}[thm]{Lemma}

\theoremstyle{definition}

\newtheorem{rmk}{Remark}[section]

\newtheorem*{claim*}{Claim}

\newcommand{\N}{\mathbb{N}}
\newcommand{\Z}{\mathbb{Z}}
\newcommand{\R}{\mathbb{R}}
\newcommand{\C}{\mathbb{C}}

\newcommand{\Sb}{\mathbb{S}}

\renewcommand{\mod}{\,\operatorname{mod}\,}

\newcommand{\intbrr}[1]{\llbracket#1\rrbracket}

\newcommand{\E}{\mathbf{E}}
\newcommand{\oneb}{\mathbf{1}}
\renewcommand{\P}{\mathbf{P}}

\newcommand{\Tr}{\operatorname{Tr}}

\newcommand{\dsim}{\sim}

\newcommand{\hs}{\mathrm{hs}}

\DeclareMathAlphabet{\mathcalboondox}{U}{BOONDOX-calo}{m}{n}
\newcommand{\utwo}{u}

\newcommand{\U}{\mathrm{U}}

\newcommand\numberthis{\stepcounter{equation}\tag{\theequation}}
\numberwithin{equation}{section}

\begin{document}

\title[]{Log-depth entanglement transition in hypercube Gaussian boson sampling} 

\author{Laura Shou$^{1,2}$}
\author{Alexey V. Gorshkov$^{1,2}$}

\address{\normalfont$^1$Joint Quantum Institute, Department of Physics, NIST/University of Maryland, College Park, MD 20742, USA}
\address{\normalfont$^2$Joint Center for Quantum Information and Computer Science,
NIST/University of Maryland, College Park, MD, 20742, USA}

\begin{abstract}
We study an entanglement transition at logarithmic depth in random hypercube linear optical networks. Starting from an initial Gaussian state with all modes squeezed, we prove that random hypercube linear optical networks on $n$ modes generate ensemble-averaged subsystem entanglement within a constant factor of its maximum in all subsystems at circuit depth $\Theta(\log n)$. 
Below depth $\log_2n$, there is an extensive subsystem with no entanglement. 
The entanglement generation in logarithmic depth can be compared to recent progress towards $O(\log n)$ depth average-case sampling hardness for Gaussian boson sampling in a hypercube network. 
\end{abstract}

\maketitle

\section{Introduction}

Gaussian boson sampling \cite{aa,hamilton2017gaussian} is expected to be a classically hard sampling problem which can be used to demonstrate quantum advantage.
It consists of preparing a squeezed Gaussian photonic state on $n$ modes, sending it into a linear optical network consisting of beamsplitters and phaseshifters, and then measuring the output state in the photon-number basis. 
The output photon-number probabilities involve hafnians, which are relatives of permanents, of complex matrices, and which are classically hard to compute exactly in the worst case \cite{valiant1979complexity}.
The sampling complexity of Gaussian boson sampling is affected by the depth and geometry of the linear optical network, with a circuit which is too shallow with poor connectivity being classically easy to simulate \cite{qi2022efficient,oh2022classical} and having low bipartite entanglement \cite{shou2026entanglement}. On the other hand, many large depth random linear optical networks are expected to look Haar random, which implies large entanglement \cite{iosue2023page,shou2026entanglement}, and for which Gaussian boson sampling is expected to be classically hard to simulate \cite{aa,hamilton2017gaussian}.
We are interested in understanding where the transition in entanglement and sampling complexity occur as a function of depth and circuit geometry, and particularly in the relationship between the two quantities.
In some systems, entanglement and complexity go together \cite{ghosh2023complexity}, but this is not always the case, and a general understanding of their relationship in various systems remains unclear.

In this paper, we study an entanglement transition in hypercube-geometry random linear optical networks, starting from an initial Gaussian state with all modes squeezed. In this setting, we find that these circuits generate large average-case (i.e.~ensemble-averaged) subsystem entanglement in all subsystems in depth $\Theta(\log n)$, while having an extensive subsystem with zero bipartite entanglement at depth $<\log_2n$.
The depth $\Theta(\log n)$ can be compared with the partial average-case sampling hardness result at depth $O(\log n)$ for this circuit from \cite{go2026computational}.

Specifically, we consider two types of hypercube linear optical networks on $n=2^p$ modes, $p\in\N$. The first we consider is the ``butterfly'' hypercube circuit (Figure~\ref{fig:butterfly}). We evolve an initial Gaussian state by a random butterfly circuit consisting of $2\times2$ Haar random beamsplitter-phaseshifter combinations, and show that the resulting average-case (i.e.~ensemble-averaged) subsystem entanglement reaches within a constant factor of its maximum possible value in all subsystems at circuit depth $\Theta(\log n)$. For depth $<\log_2n$, not all modes are connected by the circuit and so there is a subsystem with no entanglement.
To investigate the relationship between entanglement and sampling complexity, we consider a second type of hypercube circuit. Average-case sampling hardness results for shallow-depth linear optical networks are complicated due to lack of symmetry over outcomes, which prevents using the standard reduction with Stockmeyer's algorithm \cite{aa,go2026computational}. 
To address this issue, \cite[Theorem 7]{go2026computational} considers the ``kaleidoscope'' hypercube circuit (Figure~\ref{fig:kaleidoscope}), whose structure allows one to implement a random permutation to randomize the outcomes, and proves a partial\footnote{More precisely, Ref.\ \cite{go2026computational} proves average-case sampling hardness with a strict additive error allowance, and notes that it remains to improve the additive imprecision to obtain the desired classical intractability results.} average-case sampling hardness result for random kaleidoscope circuits at depth $O(\log n)$. 
We match this with entanglement generation for random kaleidoscope circuits at depth $\Theta(\log n)$.

\begin{figure}[htb]
\begin{tikzpicture}[scale=.5]
\def\dropfirst#1#2\relax{#2}
\def\rcol{gray!30}
\def\spacing{2.2cm}
\foreach \i in {0,...,7}{ 
	\node[left] at (-.2,-\i) {\pgfmathparse{bin(\i+8)}{\expandafter\dropfirst\pgfmathresult\relax}}; 
}
\foreach \l in {0,1,1.8}{ 
\draw[fill=\rcol,color=\rcol,xshift={\l*\spacing}] (-.25,.5) rectangle (1.5-.5*\l,-7.5); 
}
\begin{scope}[{Circle}-{Circle},shorten >=-2pt, shorten <=-2pt, line width =.8] 
\foreach \i in {0,...,3}{ 
	\draw (.4*\i,-\i)--++(0,-4); 
	\ifnum\i<2
		\draw[xshift=\spacing] (.4*\i+.2,-\i)--++(0,-2); 
	\else
		\draw[xshift=\spacing] ({.4*(\i-2)+.2},-\i-2)--++(0,-2);
	\fi
	\draw[xshift=2*\spacing] (-.25,-2*\i)--++(0,-1); 
}
\end{scope}
\foreach\mode in {0,...,7}{
\draw (-.3,-\mode)--(4.6,-\mode);
}
\draw[->] (-.3,-8)--(4.6,-8) node[right] {$t$};
\end{tikzpicture}
\caption{A single block of the butterfly circuit with $n=2^3=8$ modes, consisting of $p=\log_2n=3$ layers.
Each vertical line segment denotes a beamsplitter-phaseshifter gate applied across the two endpoint modes.
The circuit depth is the number of layers, i.e. for the single block here the depth is $p=\log_2n=3$.
A butterfly circuit of depth $\ell\log_2n$ then repeats this single block's structure $\ell$ times. 
}\label{fig:butterfly}
\end{figure}
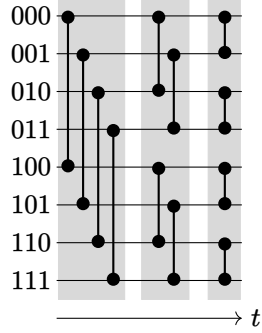

The key property of the kaleidoscope circuit used in the sampling complexity argument in \cite{go2026computational} is the following permutation implementation:
A single kaleidoscope block follows the Bene\v{s} network structure \cite{benes1965mathematical}, which is able to classically route permutations in the network without congestion. This implies any permutation can be realized in a single kaleidoscope block using a combination of identity and SWAP gates $\begin{pmatrix}0&1\\1&0\end{pmatrix}$. In contrast, this does not hold for the butterfly block structure, which can have large congestion. 

\begin{figure}[htb]
\begin{tikzpicture}[scale=.5]
\def\dropfirst#1#2\relax{#2}
\def\rcol{gray!30}
\def\spacing{2.2cm}
\foreach \i in {0,...,7}{
	\node[left] at (-.2,-\i) {\pgfmathparse{bin(\i+8)}{\expandafter\dropfirst\pgfmathresult\relax}}; 
}
\newcommand{\rectangles}{%
\foreach \l in {0,1,1.8}{
\draw[fill=\rcol,color=\rcol,xshift={\l*\spacing}] (-.25,.5) rectangle (1.5-.5*\l,-7.5); 
}
}
\rectangles
\begin{scope}[xscale=-1,xshift=-10cm]
\rectangles
\end{scope}
\begin{scope}[{Circle}-{Circle},shorten >=-2pt, shorten <=-2pt, line width =.8] 
\foreach \i in {0,...,3}{
	\draw (.4*\i,-\i)--++(0,-4); 
	\draw (.4*\i+8.75,-\i)--++(0,-4); 
	\ifnum\i<2
		\draw[xshift=\spacing] (.4*\i+.2,-\i)--++(0,-2); 
		\draw[xshift=\spacing] (.4*\i+5,-\i)--++(0,-2); 
	\else
		\draw[xshift=\spacing] ({.4*(\i-2)+.2},-\i-2)--++(0,-2);
		\draw[xshift=\spacing] ({.4*(\i-2)+5},-\i-2)--++(0,-2); 
	\fi
	\draw[xshift=2*\spacing] (-.25,-2*\i)--++(0,-1); 
	\draw[xshift=2*\spacing] (1.5,-2*\i)--++(0,-1); 
}
\end{scope}
\foreach\mode in {0,...,7}{
\draw (-.3,-\mode)--(10.3,-\mode);
}
\draw[->] (-.3,-8)--(10.3,-8) node[right] {$t$};
\end{tikzpicture}
\caption{A single block of the kaleidoscope circuit with $n=2^3$ modes, which consists of a butterfly block followed by an inverted butterfly block. 
Each connection denotes a beamsplitter-phaseshifter combination applied across the two endpoint modes.
The total number of layers, or circuit depth, for this single block is $2p=2\log_2n=6$.
A kaleidoscope circuit of depth $2\ell\log_2n$ then repeats this block's structure $\ell$ times.
}\label{fig:kaleidoscope}
\end{figure}
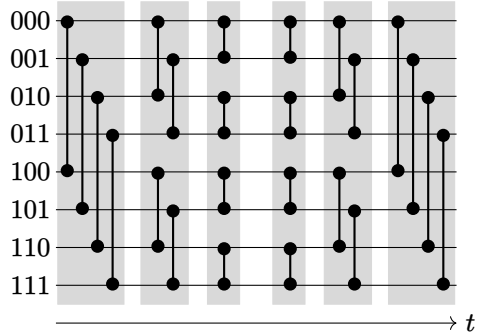

\subsection{Set-up}

We consider passive linear optical networks consisting of phaseshifters and 2-mode all-to-all beamsplitters, which are arranged in some circuit geometry. 
For studying entanglement, we start with an initial Gaussian state 
consisting of $n$ single-mode squeezed vacuum states all with squeezing parameter $s>0$.
Although we consider equal squeezing, due to \cite{shou2026anticoncentration}, our bounds also apply to more general squeezing parameters $(s_1,\ldots,s_n)\in\R^n$ by replacing $s$ with $s_\mathrm{min}:=\min_i|s_i|$.
The initial Gaussian state is then evolved according to a passive linear optical network described by an $n\times n$ linear optical unitary $U$. For a subsystem $\Gamma\subset\{0,\ldots,n-1\}$, the R\'enyi-2 entropy of the reduced evolved state $\rho_\Gamma(U)$ is $S_2(U)=-\log\Tr \rho_\Gamma(U)^2$, which can also be written (since the state is Gaussian) as \cite{Serafini2017Quantum-Continu}
\begin{align}
S_2(U)&=\frac12\log\det\sigma(U)=\frac12\Tr\log\sigma(U),
\end{align}
where $\sigma(U)$ is the covariance matrix of $\rho_\Gamma(U)$.

We are interested in two hypercube-geometry linear optical networks, the butterfly circuit (Figure~\ref{fig:butterfly}) and the kaleidoscope circuit (Figure~\ref{fig:kaleidoscope}).
In both cases, we have $n=2^p$ modes $\intbrr{0:2^p-1}:=\{0,\ldots,2^p-1\}$, $p\in\N$. For mode $i\in\intbrr{0:2^p-1}$, we consider its binary expansion expressed using $p$ bits. For two bit strings $i$ and $r$ of the same length, we let $i\oplus r$ denote bitwise XOR.
Let $e_t$ be the string of bits with a one in the $t$th position as counted from the left ($t\in\{1,\ldots,p\}=\intbrr{1:p}$), and zeros in all $p-1$ other positions, and let $0\cdots0$ denote the string of $p$ zero bits. We abbreviate $r_t\in\{0\cdots0,e_t\}$ as $r_t\in\{0,1\}$, and will write for example, $i\oplus r_t$ for $r_t\in\{0,1\}$ to mean a bit flip in the $t$th position if $r_t=1$, and no flip (identity) if $r_t=0$.

A single block of the butterfly hypercube circuit (Figure~\ref{fig:butterfly}) consists of $p$ layers, with the $t$th layer consisting of $n/2$ gates applied in parallel, which connect modes $i$ and $i\oplus e_t$. 
Each connection drawn corresponds to applying a beamsplitter-phaseshifter combination across the connected modes. 
In terms of the $n\times n$ linear optical unitary $U$, this corresponds to applying a $2\times2$ unitary matrix in the $(i,i\oplus e_t)$ coordinate plane.
The butterfly circuit then repeats the single block structure of Figure~\ref{fig:butterfly}, allowing for different $\U(2)$ matrices for the different gates and blocks. The depth of the circuit is the number of total layers, so for example a single block has depth $p=\log_2n$.

The kaleidoscope circuit is a hypercube circuit constructed from two types of butterfly blocks. A single block for the kaleidoscope circuit is constructed as a butterfly block followed by an inverted butterfly block (Figure~\ref{fig:kaleidoscope}). 
This block then consists of $2p$ layers, where the $t$th layer for $t=1,\ldots,p$ connects modes $(i,i\oplus e_t)$, and for $t=p+1,\ldots,2p$ connects modes $(i,i\oplus e_{2p-t+1})$. The kaleidoscope circuit then repeats this structure. 
As before, the depth of the circuit is the number of individual layers, so for example a single block for the kaleidoscope has depth $2p=2\log_2n$.

\subsection{Main results}
Our main results are the following on entanglement generation in depth $\Theta(\log n)$, for the random butterfly circuit and random kaleidoscope circuit.
Throughout, we always consider a random circuit to mean every beamsplitter-phaseshifter gate is drawn independently at random from Haar measure on $\U(2)$.
We state the results precisely for R\'enyi-2 entropy, but similar bounds hold for the von Neumann entropy and the R\'enyi-$\alpha$ entropies for $\alpha$ an integer $\ge2$ (see Remark~\ref{rmk:1}(ii)).

\begin{thm}[butterfly circuit entanglement]\label{thm:butterfly-s2}
Consider a depth $T$ random butterfly circuit on $n=2^p$ modes, $p\in\N$, applied to an initial Gaussian state consisting of single-mode squeezed states with equal squeezing parameters $s>0$.
For $T\ge 20\log_2n$, the average-case R\'enyi-2 entropy for any subsystem $\Gamma\subset\{0,\ldots,n-1\}$ satisfies
\begin{align}\label{eqn:as2}
\E S_2(U)&\ge\frac{|\Gamma|(n-|\Gamma|)}{n}\log(\cosh(2s))(1-\varepsilon_n),
\end{align}
where $\varepsilon_n=6n^{-1/8}=o(1)$.
Additionally, for depth $T<\log_2 n$, there is a subsystem (of size $n/2$) with zero entanglement. Therefore, $\Theta(\log n)$ is the precise order of this entanglement transition.
\end{thm}
\begin{rmk}
\phantomsection\label{rmk:1}
\begin{enumerate}[(i)]
\item The maximum possible R\'enyi-2 entropy for $\Gamma$ with the initial Gaussian state is $\min(|\Gamma|,n-|\Gamma|)\log\cosh(2s)$, so for e.g. $|\Gamma|\le n/2$, \eqref{eqn:as2} is smaller by a factor $1-|\Gamma|/n$, along with the $1-\varepsilon_n$ factor. 
The Haar random expectation values, corresponding to when $U$ is an $n\times n$ Haar random unitary, can also be a factor smaller than the maximum value for extensive subsystems $\Gamma$ \cite{iosue2023page}; this difference between the maximum and average entropy is called the \emph{Page correction}. We compare the exact limiting expressions across different subsystem sizes $r=|\Gamma|/n$ for these quantities with squeezing $s=3/4$ in Figure~\ref{fig:page}. The bound \eqref{eqn:as2} follows the general Haar Page curve shape, and under normalization by $1/n$ it differs at $r=1/2$ by $\frac12\log\cosh s-\frac14\log(1+\tanh^2s)$ from the Haar value.

We expect the discrepancy between the bound \eqref{eqn:as2} and the limiting Haar value comes from loose estimates used for higher powers in a power series expansion of $S_2(U)$ [Eq.~\eqref{eqn:s2lower}]. For example, when $s\to0$, the lowest term in the series dominates, and we see the difference from the Haar value at $r=1/2$ tends to zero.

Due to the above, with the subsequent proof method, we only ensure asymptotic Haar values for $\E S_2(U)$ by taking much larger depth as follows. As a corollary of the theorem, we can obtain that the full butterfly linear optical unitary $U$ itself becomes $\varepsilon$-close to Haar measure on $\U(n)$ in $L^2$ Wasserstein distance in depth $O(n\log n\log(n/\varepsilon))$ by the proof of \cite[Theorem 4.5(ii)]{shou2026entanglement}. This implies convergence of $\E S_2(U)$ to Haar values as well as (weak) Haar typicality for $S_2(U)$ at these depths, but we expect such large depth is unnecessary for Haar behavior of $\E S_2(U)$.

\item As mentioned above, due to an unequal squeezing result in \cite{shou2026anticoncentration}, the bound \eqref{eqn:as2} also applies for general squeezing parameters $(s_1,\ldots,s_n)$ by replacing $s$ with $s_\mathrm{min}:=\min_i|s_i|$.

\item Since the von Neumann entropy $S_1$ satisfies $S_1\ge S_2$, we immediately obtain the lower bound \eqref{eqn:as2} for $\E S_1(U)$. The lower bound \eqref{eqn:as2} also holds for $\E S_\alpha(U)$ any $0<\alpha<2$ by monotonicity in $\alpha$.
For the R\'enyi-$\alpha$ entropies $S_\alpha(U)$, $\alpha\ge2$, similar lower bounds follow using \cite[Eq.~(B11)]{youm2025average} and the bound \eqref{eqn:uutlower} in the proof of Theorem~\ref{thm:butterfly-s2}: For $T\ge20\log_2n$ and $\alpha\ge2$ an integer,
\begin{align}\label{eqn:salpha}
\E S_\alpha(U)&\ge \left[\oneb_{\alpha\in 2\Z}\log\cosh(2s)+\sum_{\ell=1}^\infty\sum_{m=1}^{\lfloor\frac{\alpha-1}{2}\rfloor}\frac{\sinh^{2\ell}(2s)}{\ell(\cosh^2(2s)+\cot^2\frac{\pi m}{\alpha})^\ell}\right]\frac{|\Gamma|(n-|\Gamma|)}{(\alpha-1)n}(1-\varepsilon_n).
\end{align}
Bounds for real $\alpha\in[2,\infty)$ then follow using monotonicity of R\'enyi-$\alpha$ entropies in $\alpha$ to compare to the bound for $\lceil\alpha\rceil\in\Z$.

\item The error bound $\varepsilon_n=6n^{-1/8}$ can be made a bit smaller by taking a larger depth $T$. For $T\ge 27\log_2n$, one can take $\varepsilon_n=6n^{-1}$.
\end{enumerate}
\end{rmk}

\begin{figure}[htb]
\includegraphics[width=3.25in]{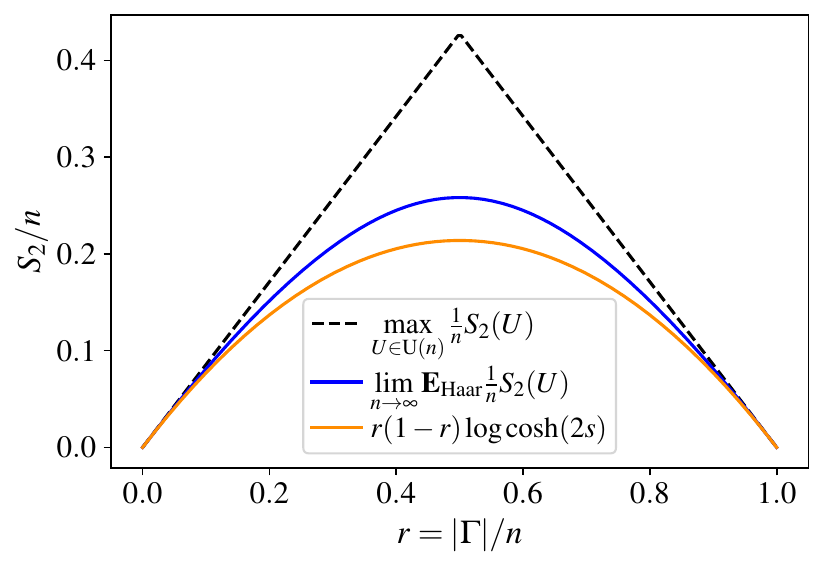}
\caption{Comparison of the (normalized) bound \eqref{eqn:as2} (in orange) with the R\'enyi-2 Page curve for Haar random unitaries (blue) from \cite{iosue2023page,shou2026anticoncentration}, for squeezing $s=3/4$ and in the limit $n\to\infty$. The maximum possible subsystem entropy (dashed) for the input Gaussian state is also shown. 
}\label{fig:page}
\end{figure}

Entanglement for the butterfly hypercube circuit was considered numerically in \cite{go2024exploring}, where the Page curves were found to numerically converge to the Haar Page curves \cite{iosue2023page} with a constant number of butterfly blocks (i.e.~total number of layers $\Theta(\log n)$), for randomly chosen subsystems of size $k$, and independent of the system size $n$. Additionally, a single block (depth $p=\log_2n$) of the random butterfly hypercube circuit appeared sufficient for the Page curve, for the instances of randomly chosen subsystems, to get within a constant factor of the Haar Page curve. However, note that for Theorem~\ref{thm:butterfly-s2}, we consider \emph{all} (or worst-case) subsystems, which could end up requiring more circuit blocks than when considering random or typical subsystems.

Instead of considering the classical graph random walk relation and recursive structure as in \cite{shou2026entanglement}, we use a different method utilizing specific properties of the hypercube circuits, based on Markov chain coupling for Kac's walk on the sphere \cite{kac1956foundations,pillai2017kacs,lu2024quantum}.
In particular, we will prove Theorems~\ref{thm:butterfly-w2} and \ref{thm:kaleidoscope-w2} on $L^2$ Wasserstein convergence of associated butterfly or kaleidoscope random walks on the complex unit sphere.
This allows for obtaining the precise order for entanglement generation in these cases.
We write explicit numerical constants on all bounds, although we do not attempt to optimize the constants (particularly not for the kaleidoscope circuit bounds below).

Next, recall we consider the kaleidoscope circuit in order to compare to the partial hardness of sampling result at depth $O(\log n)$ from \cite{go2026computational}, which utilized the special permutation implementation property of the kaleidoscope geometry. As we show, this sampling complexity result is matched by entanglement generation in depth $\Theta(\log n)$.
\begin{thm}[kaleidoscope circuit entanglement]\label{thm:kaleidoscope-s2}
Consider a depth $T$ random kaleidoscope circuit on $n=2^p$ modes, $p\in\N$, applied to an initial Gaussian state consisting of single-mode squeezed states with equal squeezing parameters $s>0$.
For $T\ge 520\log_2n$, the average-case R\'enyi-2 entropy for any subsystem $\Gamma\subset\{0,\ldots,n-1\}$ satisfies
\begin{align}\label{eqn:as2-kaleidoscope}
\E S_2(U)&\ge\frac{|\Gamma|(n-|\Gamma|)}{n}\log(\cosh(2s))(1-\varepsilon_n),
\end{align}
where $\varepsilon_n=(1+4\cdot 2^{1/4}(\log_2n)^{52})n^{-1}=o(1)$.
Additionally, for depth $T<\log_2 n$, there is a subsystem (of size $n/2$) with zero entanglement. Therefore, $\Theta(\log n)$ is the precise order of this entanglement transition.
\end{thm}

\begin{rmk}
\begin{enumerate}[(i)]

\item As for Theorem~\ref{thm:butterfly-s2}, one can extend to average-case R\'enyi-$\alpha$ entropies, and unequal squeezing in terms of $s_\mathrm{min}=\min_i|s_i|$.

\item Due to its similarity to the butterfly circuit, we would naturally expect the kaleidoscope circuit to have very similar entanglement properties, as we describe in the paragraph below.
The proof ends up being more complicated due to the repeated layers in a single block, which complicate the coupling argument bounds. In particular, we do not obtain a clean contraction property such as Lemma~\ref{lem:contraction} which is used for the butterfly circuit, and instead will prove an iterative bound Lemma~\ref{lem:cycle-iteration}. This method ends up producing a larger depth requirement and larger error bounds, even though we expect that the kaleidoscope entanglement convergence should be more similar to the butterfly's.

We give an intuitive explanation for comparing the random butterfly and kaleidoscope circuits.
If comparing by circuit depth $T$, we expect that the butterfly circuit should mix slightly faster than the kaleidoscope circuit, because it cycles through different layers better (the middle and beginning of a kaleidoscope block repeat recently applied layers).
Additionally, we expect a single kaleidoscope block (length $2p$) should mix better than a single butterfly block (length $p$) since it contains the butterfly block. 
Overall, this suggests that the kaleidoscope circuit with $m$ blocks ($2mp$ depth) should mix faster than butterfly with $m$ blocks ($mp$ depth), but slower than butterfly with $2m$ blocks ($2mp$ depth).
\end{enumerate}
\end{rmk}

Theorems~\ref{thm:butterfly-s2} and \ref{thm:kaleidoscope-s2} consider average entanglement, though we also obtain some brief typicality results, particularly for subsystems of size $|\Gamma|=o(n)$, which we describe in Section~\ref{sec:typicality}.

\subsection{Outline and notation}
The rest of the paper is organized as follows.
\begin{itemize}
\item In Section~\ref{sec:butterfly}, we discuss and prove Theorem~\ref{thm:butterfly-s2} for the butterfly circuit.
\item In Section~\ref{sec:kaleidoscope}, we prove Theorem~\ref{thm:kaleidoscope-s2} for the kaleidoscope circuit.
\item In Section~\ref{sec:typicality}, we prove some brief typicality results.
\end{itemize}

We collect some notation we use in the paper:
\begin{itemize}
\item $i\oplus r$: bitwise XOR between bit strings $i$ and $r$
\item $\E$ and $\P$: expectation value and probability, respectively, over the random linear optical unitary
\item $X\dsim\nu$: random variable $X$ is distributed according to the probability measure $\nu$
\item $\mathcal L(X)$: law of the random variable $X$
\item $\Sb^{n-1}=\{x=(x_1,\ldots,x_n)\in\C^n:\|x\|_2=1\}$
\item $\mu_n$: Haar measure on the sphere $\Sb^{n-1}$
\item $\U(n)$: unitary group of $n\times n$ unitary matrices
\item $\N=\{1,2,\ldots\}$
\item $o,O,\Theta,\Omega,\omega$: standard asymptotic notation
\end{itemize}

\section{Butterfly circuit}\label{sec:butterfly}

In this section, we prove Theorem~\ref{thm:butterfly-s2}. 
To do this, we will prove a mixing property for a corresponding random butterfly walk on the complex unit sphere $\Sb^{n-1}$. 
The mixing will imply that a single column of the random linear optical unitary $U$ is close to a Haar random unit vector. This will allow us to bound the average-case entanglement for any subsystem, as this will involve fourth moments of entries of $U$ from a single column, which we can bound using the similarity to a Haar random unit vector (Eqs.~\eqref{eqn:s2entropy}--\eqref{eqn:uutlower}).

For an initial state $|X_0\rangle\in\Sb^{n-1}$, consider the discrete-time walk defined via $|X_t\rangle=\utwo^{[t]}|X_{t-1}\rangle$, where $\utwo^{[t]}$ is a random $n\times n$ unitary corresponding to the $t$th layer of the random butterfly circuit, i.e. $\utwo^{[t]}$ implements $n/2$ gates in parallel across mode pairs $(i,i\oplus e_{t})$, with all $2\times 2$ gates drawn independently from Haar measure and independently of all previous $\utwo^{[s]}$'s. 
Here we take the time index $t$ in $e_t$ to be modulo $p$ and in $\intbrr{1:p}$.
Note that if we let $U=U_t=\prod_{s=t}^1 u^{[s]}$ be the depth $t$ random butterfly circuit with all $u^{[s]}$ independent, then we can write $|X_t\rangle\dsim U_t|X_0\rangle$.

The $L^p$ Wasserstein distance with respect to the Euclidean norm $\|\cdot\|_2$ on $\Sb^{n-1}$ is defined as $W_p(\nu,\mu)=\inf\{\E[\|X-Y\|_2^p]^{1/p}:(X,Y)\text{ is a coupling of $(\nu,\mu)$}\}$, where $(X,Y)$ is said to be a coupling of $(\nu,\mu)$ if the marginals are distributed as $X\dsim \nu$ and $Y\dsim\mu$. For $p=1$, there is the dual characterization,
\begin{align}\label{eqn:w1dual}
W_1(\nu,\mu)=\sup\left\{\int_{\Sb^{n-1}}f\,d(\nu-\mu):f:\Sb^{n-1}\to\R\text{ is 1-Lipschitz}\right\}.
\end{align}
Note that $W_1(\nu,\mu)\le W_p(\nu,\mu)$ for $p\ge1$.
We are interested in the case where $\mu$ is the Haar measure on $\Sb^{n-1}$ and $\nu$ is the law $\mathcal L(X_t)$ of the butterfly random walk $X_t$. Then $\varepsilon$-closeness in $W_1$ distance means that expectation values of $1$-Lipschitz functions of $X_t$ are $\varepsilon$-close to the $\Sb^{n-1}$ Haar value.
The following convergence result will be used to prove Theorem~\ref{thm:butterfly-s2}.

\begin{thm}[butterfly walk on the sphere]\label{thm:butterfly-w2}
The butterfly random walk $X_t$ on the sphere $\Sb^{n-1}$ mixes in $L^2$ Wasserstein distance in time $T=\Theta(\log n)$. More precisely, for any initial distribution $X_0\sim\nu$ on the sphere, and any $\kappa\ge1$ and $T\ge (8\kappa+3)\log_2n$,
\begin{align}
W_2(\mathcal L(X_T),\mu_n)&\le \frac{2^{1/4}}{n^\kappa},
\end{align}
where $\mu_n$ denotes Haar measure on $\Sb^{n-1}$.
\end{thm}
\begin{rmk}
One should be able to upgrade the Wasserstein distance to a result on total variation distance if desired, using the second phase coupling in \cite{pillai2017kacs,lu2024quantum}. However, this will not be necessary for the results here.
\end{rmk}

We discuss the connections of Theorem~\ref{thm:butterfly-w2} to Kac's random walk on the sphere. The usual, or sequential, Kac's random walk on the complex $n$-sphere $\Sb^{n-1}$ is a discrete-time Markov chain $\{X_t\in\Sb^{n-1}\}_{t}$ which, for one time step, chooses a coordinate pair $(i,j)$ with $i\ne j$ uniformly at random and applies an independent Haar random $2\times2$ unitary matrix across those coordinates.
It was shown that sequential Kac's walk on the (real) sphere mixes (in total variation distance) in time $O(n\log n)$ \cite{pillai2017kacs}, and that parallel Kac's walk on the sphere, which in one step chooses a random pairing of the modes $\intbrr{0:n-1}$ and applies $n/2$ independent Haar random $\operatorname{SU}(2)$ unitaries across each of the coordinate pairs, mixes in time $O(\log n)$ \cite{lu2024quantum}. However, the proofs do not readily imply mixing for any particular sequence of coordinate plane rotations, or even mixing with high probability over the possible coordinate plane sequences. Additionally, subsystem entanglement averaged over Kac's walk itself is not so meaningful for our purposes; for a given subsystem $\Gamma$, we want to understand the entanglement generation for a single circuit geometry, not an average over arbitrary geometries, many of which will rapidly produce large entanglement in the specific subsystem $\Gamma$ but not in all subsystems. 

Therefore, we want to exhibit a specific circuit geometry which generates large average-case entanglement in all subsystems in depth $O(\log n)$.
Intuitively, there should be many sequences of coordinate plane rotations with walks which converge rapidly in time $O(\log n)$ like parallel Kac's walk, such as sequences that mix all $n$ modes rapidly in an expander-like geometry. However, without the uniform averaging of Kac's walk over the modes, it becomes difficult in general to extract precise bounds.
We find that the symmetry and sparseness of the butterfly hypercube circuit allows one to establish a coupling contraction estimate similar to that for Kac's random walk. For the kaleidoscope circuit (Section~\ref{sec:kaleidoscope}), which simply consists of a butterfly circuit followed by an inverted butterfly circuit, the analysis is already more complicated.

To prove Theorem~\ref{thm:butterfly-w2}, we will use the proportional coupling used for Kac's random walk \cite{pillai2017kacs} and also adapted for $\operatorname{SU}(n)$ evolution in \cite{lu2024quantum}. To define the coupling, which will be for $\U(n)$ evolution, parametrize a $2\times2$ unitary matrix as
\[
u(\alpha,\beta,\theta,\phi)=e^{i\phi}\begin{bmatrix}e^{i\alpha}\cos\theta&-e^{i\beta}\sin\theta\\
e^{-i\beta}\sin\theta&e^{-i\alpha}\cos\theta\end{bmatrix}.
\]
This is Haar random if $\alpha,\beta,\phi\in[0,2\pi)$ are independent uniform random and $\theta=\arcsin\sqrt{\zeta}$ where $\zeta\in[0,1)$ is independent and uniform random.
The proportional coupling for rotation in the coordinate plane $(i,j)$, $i<j$, lines up the directions of $x=(X[i],X[j])$ and $y=(Y[i],Y[j])$ as follows.
As in \cite{lu2024quantum}, let $l_1=\sqrt{|X[i]|^2+|X[j]|^2}$ and $l_2=\sqrt{|Y[i]|^2+|Y[j]|^2}$, and define $2\times2$ unitaries $u_0,v_0$ such that\footnote{If $l_1=0$, then one may take any fixed or independent $2\times2$ unitary for $u_0$, and similarly for $v_0$ if $l_2=0$.}
\begin{align}
u_0\begin{pmatrix}X[i]\\X[j]\end{pmatrix}&=\begin{pmatrix}l_1\\0\end{pmatrix},\quad v_0\begin{pmatrix}Y[i]\\Y[j]\end{pmatrix}=\begin{pmatrix}l_2\\0\end{pmatrix}.
\end{align}
After this, apply the same Haar random unitary $u(\alpha,\beta,\theta,\phi)$ to both to obtain the next steps 
\begin{align}\label{eqn:ijcoupling}
x'=u(\alpha,\beta,\theta,\phi)u_0x,\qquad y'=u(\alpha,\beta,\theta,\phi)v_0y.
\end{align}
Since the distribution of $u(\alpha,\beta,\theta,\phi)$ is invariant under multiplication by $u_0$ or $v_0$, this is a coupling for a one-step random rotation in the $(i,j)$ plane. 

For one layer, or time-step, of the butterfly random walk $X_t$, we apply $n/2$ gates in parallel, which act on disjoint modes. The proportional coupling for $(X_t,Y_t)$ for a single time-step is simply \eqref{eqn:ijcoupling} applied in each of the coordinate planes $(i,i\oplus e_t)$, $i<i\oplus e_t$, with independent $\alpha_i,\beta_i,\theta_i,\phi_i$.
As before, we take the subscript $t$ in $e_t$ to be modulo $p$ and in $\intbrr{1:p}$. For convenience, we may refer to $i\oplus e_t$ as a bit flip in the $t$th bit, meaning that it is a bit flip in the $(t\mod p)\in\intbrr{1:p}$th bit.
For this proportional coupling, we will show
\begin{lem}[contraction for butterfly circuit]\label{lem:contraction}
Consider the depth $t$, $n=2^p$-mode butterfly linear optical unitary $U_t$, and the associated random walk on the sphere, $U_t|\psi_0\rangle\in\Sb^{n-1}$.
Fix $X_0,Y_0\in\Sb^{n-1}$. Let $(X_t,Y_t)$ be the proportional coupling defined above for the random walk $U_t|\psi_0\rangle$, and let $\delta_t[i]:=|X_t[i]|^2-|Y_t[i]|^2$ for modes $i\in\{0,\ldots,n-1\}$.
Then for $t\ge p$,
\begin{align}\label{eqn:butterfly-c}
\E[\|\delta_t\|_2^2]\le\frac23\E[\|\delta_{t-1}\|_2^2].
\end{align}
\end{lem}
\begin{proof}
Let $u(\alpha_i,\beta_i,\theta_i,\phi_i)$ denote the Haar random $2\times 2$ unitary in \eqref{eqn:ijcoupling} [denoted by $u(\alpha,\beta,\theta\phi)$ there] which is applied across modes $(i,i\oplus e_t)$, with $i<i\oplus e_t$.
The proportional coupling for $(X_t,Y_t)$ is, for $i\in\{0,\ldots,2^p-1\}$ with $i<i\oplus e_t$,
\begin{align}
\begin{aligned}
X_t[i]&=e^{i\phi_i}e^{i\alpha_i}\sqrt{|X_{t-1}[i]|^2+|X_{t-1}[i\oplus e_t]|^2}\cos\theta_i,\; X_t[i\oplus e_t]=e^{i\phi_i}e^{-i\beta_i}\sqrt{|X_{t-1}[i]|^2+|X_{t-1}[i\oplus e_t]|^2}\sin\theta_i,\\
Y_t[i]&=e^{i\phi_i}e^{i\alpha_i}\sqrt{|Y_{t-1}[i]|^2+|Y_{t-1}[i\oplus e_t]|^2}\cos\theta_i,\;
Y_t[i\oplus e_t]=e^{i\phi_i}e^{-i\beta_i}\sqrt{|Y_{t-1}[i]|^2+|Y_{t-1}[i\oplus e_t]|^2}\sin\theta_i,
\end{aligned}\label{eqn:coupling4}
\end{align}
where the time index $t$ in $e_t$ is taken modulo $p$ and in $\intbrr{1:p}$.
One can check that $\E\cos^2\theta=\E\sin^2\theta=1/2$, and that $\E\cos^4\theta=\E\sin^4\theta=1/3$, for $\theta=\arcsin\sqrt{\zeta}$ with $\zeta\dsim\operatorname{Uniform}[0,1)$.
Let $\mathcal F_{t-1}$ denote the $\sigma$-algebra generated by $X_s$ and $Y_s$ up to time $s=t-1$.
For applying all the gates in layer $t$, note that for a given mode $i$, only the single gate connecting $i$ and $i\oplus e_t$ affects this mode. Then using \eqref{eqn:coupling4}, we get
\begin{align}
\E[\delta_t[i]|\mathcal F_{t-1}]&=\frac12(\delta_{t-1}[i]+\delta_{t-1}[i\oplus e_t]),\label{eqn:split}\\
\E[\delta_t[i]^2|\mathcal F_{t-1}]&=\frac13(\delta_{t-1}[i]+\delta_{t-1}[i\oplus e_t])^2.\label{eqn:split2}
\end{align}
From the second line, we then have
\begin{align*}
\E[\|\delta_t\|_2^2|\mathcal F_{t-1}]
&=\frac23\|\delta_{t-1}\|_2^2+\frac23\sum_{i=0}^{2^p-1}\delta_{t-1}[i]\delta_{t-1}[i\oplus e_t].\numberthis\label{eqn:iter1}
\end{align*}
We claim 
\begin{align}\label{eqn:claim0}
\E[\delta_{t-1}[i]\delta_{t-1}[i\oplus e_{t}]]\le 0,
\end{align}
for any mode $i\in\{0,\ldots,n-1\}$ and any $t\ge p$.
To see this, we iterate back in time $p-1$ times. The $s$th layer links modes which differ exactly on the bit $(s\mod p)\in\intbrr{1:p}$ from the left.
Note that $i$ and $i\oplus e_t$ differ on the $t$th bit, and we do \emph{not} apply a layer for the $t$th bit in the next $p-1$ steps, since we just applied one at time $t$ in \eqref{eqn:iter1}. 
As a result, when calculating expectations of the form $\E[\delta_{t-k}[i']\delta_{t-k}[j']|\mathcal F_{t-k-1}]$ below,  $i'$ and $j'$ always differ on the $t$th bit. When we apply an $s$th layer with $s\ne t$, the two relevant gates are $(i',i'\oplus e_s)$ and $(j',j'\oplus e_s)$, which are distinct and independent, and so we can evaluate the conditional expectations over the two $\delta_{t-k}[\cdot]$ terms separately using \eqref{eqn:split}.
Using this observation, and taking the $j$ subscript indices in $r_j,r_j'$ below modulo $p$ and in $\intbrr{1:p}$, we have
\begin{align*}
&\E[\delta_{t-1}[i]\delta_{t-1}[i\oplus e_{t}]]\\
&=\frac1{2^2}\sum_{r_{t-1},r'_{t-1}\in\{0,1\}}\E\delta_{t-2}[i\oplus r_{t-1}]\delta_{t-2}[i\oplus e_{t}\oplus r'_{t-1}]\\
&=\frac1{2^{2(p-1)}}\sum_{r_{t-1},r'_{t-1}\in\{0,1\}}\cdots\!\!\!\sum_{r_{t-(p-1)},r'_{t-(p-1)}\in\{0,1\}}\E\delta_{t-p}[i\oplus r_{t-1}\cdots\oplus r_{t-(p-1)}]\delta_{t-p}[i\oplus e_{t}\oplus r'_{t-1}\cdots\oplus r'_{t-(p-1)}]\\
&=\E\frac1{2^{2(p-1)}}\sum_{r_{t-1},\ldots,r_{t-(p-1)}\in\{0,1\}}\delta_{t-p}[i\oplus r_{t-1}\cdots\oplus r_{t-(p-1)}]\sum_{r'_{t-1},\ldots,r'_{t-(p-1)}\in\{0,1\}}\delta_{t-p}[i\oplus e_{t}\oplus r'_{t-1}\cdots\oplus r'_{t-(p-1)}].\numberthis\label{eqn:iterate-p}
\end{align*}
The subscript indices $j$ in $r_j,r_j'$ run over  $j\in\{t-1,\ldots,t-(p-1)\}=\{1,\ldots,p\}\setminus\{t\}$, taking everything modulo $p$ and in $\intbrr{1:p}$. Letting
\begin{align*}
a_i&:=\sum_{r_{t-1},\ldots,r_{t-(p-1)}\in\{0,1\}}\delta_{t-p}[i\oplus r_{t-1}\cdots\oplus r_{t-(p-1)}],\\
b_i&:=\sum_{r'_{t-1},\ldots,r'_{t-(p-1)}\in\{0,1\}}\delta_{t-p}[i\oplus e_{t}\oplus r'_{t-1}\cdots\oplus r'_{t-(p-1)}],
\end{align*}
we see $a_i+b_i=\sum_{k=0}^{2^p-1} \delta_{t-p}[k]=0$ for any $i$. Therefore $a_ib_i\le0$, showing that \eqref{eqn:iterate-p} is $\le0$, so \eqref{eqn:claim0} holds, and \eqref{eqn:butterfly-c} follows from \eqref{eqn:iter1}. 
\end{proof}

\begin{proof}[Proof of Theorem~\ref{thm:butterfly-w2}]
Let $Y_t$ be a butterfly random walk with initial distribution $Y_0\dsim\mu_n$, where $\mu_n$ is Haar measure on the $n$-sphere; then $Y_t\dsim\mu_n$ as well by Haar measure invariance under unitary rotations.
From Lemma~\ref{lem:contraction}, we obtain for the proportional coupling $(X_t,Y_t)$ and any $t\ge p$,
\begin{align}
\E[\|\delta_t\|_2^2]&\le \left(\frac23\right)^{t-p}\E[\|\delta_{t-(t-p)}\|_2^2]\le 2\left(\frac23\right)^{t-p},
\end{align}
also using that $\|\delta_s\|_2^2=\sum_k |X_s[k]|^4+|Y_s[k]|^4-2|X_s[k]|^2|Y_s[k]|^2\le 2$ for any $s$.
Recalling that Wasserstein distance is an infimum over couplings, then just as shown for parallel Kac's walk in \cite{lu2024quantum}, we have
\begin{align*}
W_2(\mathcal L(X_t),\mu_n)&\le (\E\|X_t-Y_t\|_2^2)^{1/2}\le (n\E\|X_t-Y_t\|_4^4)^{1/4}\\ 
&=\left(n\E\sum_{i=0}^{n-1}||X_t[i]|-|Y_t[i]||^4\right)^{1/4}\le \left(n\E\sum_{i=0}^{n-1}(|X_t[i]|-|Y_t[i]|)^2(|X_t[i]|+|Y_t[i]|)^2\right)^{1/4} \\
&= (n\E\|\delta_t\|_2^2)^{1/4}\le (2n)^{1/4}\left(\frac23\right)^{(t-p)/4},\numberthis\label{eqn:w2bound} 
\end{align*}
where the equality in the second line comes from using that $X_t[i]$ and $Y_t[i]$ are coupled to have the same argument (e.g. \eqref{eqn:coupling4}) so that $|X_t[i]-Y_t[i]|=||X_t[i]|-|Y_t[i]||$.
For $t\ge (8\kappa+3)\log_2n$, one can check \eqref{eqn:w2bound} is $\le2^{1/4}/n^\kappa$.
\end{proof}

\begin{proof}[Proof of Theorem~\ref{thm:butterfly-s2}]

Let $U=U_t$ be a depth $t$ random butterfly linear optical unitary, and let $(X_t)_t$ denote the associated butterfly random walk on the sphere $\Sb^{n-1}$ as in Lemma~\ref{lem:contraction}.
Let the starting vector $X_0$ for the walk be the $\alpha$th standard basis vector, $X_0=e_\alpha$. Then $|X_t\rangle=U_t|X_0\rangle$ is the $\alpha$th column of $U_t$, so letting $U_{ij}$ denote the matrix element $\langle i|U_t|j\rangle$, we see
\begin{align*}
\E[|U_{x\alpha}|^2|U_{y\alpha}|^2]&=\E[|X_t[x]|^2|X_t[y]|^2].
\end{align*}
Theorem~\ref{thm:butterfly-w2} implies that for $t\ge (8\kappa+3)\log_2 n$, $X_t$ is $2^{1/4}n^{-\kappa}$-Wasserstein close to a Haar random unit vector in $\Sb^{n-1}$.
The function $f_{xy}(X)=|X[x]|^2|X[y]|^2$ is 4-Lipschitz using the triangle inequality and that $\|X\|_\infty\le1$. 
Take $\kappa=17/8=2.125$ so that $8\kappa+3=20$. Then for $t\ge 20\log_2n$, $Z$ a Haar random unit vector on $\Sb^{n-1}$, and $x\ne y$, Theorem~\ref{thm:butterfly-w2} and \eqref{eqn:w1dual} imply
\begin{align*}
\E[|X_t[x]|^2|X_t[y]|^2]&\ge\E_{\mu_n}[|Z[x]|^2|Z[y]|^2]-\frac{4(2^{1/4})}{n^{17/8}}\\
&=\frac1{n(n+1)}-\frac{4(2^{1/4})}{n^{17/8}}
\ge\frac{1-\varepsilon_n}{n^2},\numberthis\label{eqn:lowerbound}
\end{align*}
for $\varepsilon_n=n^{-1}+4(2^{1/4}){n^{-1/8}}\le 6{n^{-1/8}}=o(1)$. (This error term can be made a bit smaller by taking larger depth and $\kappa$.) The argument of \cite[Theorem 4.5]{shou2026entanglement} then gives \eqref{eqn:as2}. For completeness, we sketch the argument here.
From \cite{iosue2023page}, there is the power series expansion for the R\'enyi-2 entropy,
\begin{align}\label{eqn:s2entropy}
S_2(U)&=\sum_{\ell=1}^\infty\frac{\tanh^{2\ell}(2s)}{2\ell}(|\Gamma|-\Tr W^\ell),
\end{align}
for $W=\Pi UU^T\Pi\bar U\bar U^T\Pi$ with $\Pi$ the $n\times n$ projection matrix onto $\Gamma$. Since $\Tr W^\ell\le \Tr W$ as all eigenvalues of $W$ are between 0 and 1, we have
\begin{align}\label{eqn:s2lower}
S_2(U)\ge(|\Gamma|-\Tr W)\log\cosh(2s).
\end{align}
We then only need a bound on $\E\Tr W$, which is
\begin{align}\label{eqn:trw}
\E\Tr W=\E\|\Pi UU^T\Pi\|_\hs^2&=|\Gamma|-\sum_{x\in\Gamma}\sum_{y\not\in\Gamma}\E|\langle x|UU^T|y\rangle|^2.
\end{align}
By pulling out the right-most single-step factor from $U=U_t$ to write $U=Qu^{[1]}$ with $Q$ and $u^{[1]}$ independent, and $u^{[1]}$ corresponding to just a single time-step/layer of the circuit, we have
\begin{align*}
\E|\langle x|UU^T|y\rangle|^2=\sum_{\alpha,\alpha'}\E U_{x\alpha}U_{y\alpha}\bar U_{x\alpha'}\bar U_{y\alpha'}
&=\sum_{\alpha,\alpha'}\sum_{i,j,i',j'=1}^n\E[Q_{xi}Q_{yj}\bar Q_{xi'}\bar Q_{yj'}]\E[u^{[1]}_{i\alpha}u^{[1]}_{j\alpha}\bar u^{[1]}_{i'\alpha'}\bar u^{[1]}_{j'\alpha'}]\\
&=\sum_\alpha\E|U_{x\alpha}|^2|U_{y\alpha}|^2,\numberthis
\end{align*}
since terms with $\alpha\ne\alpha'$ are zero, as nonzero entries of $u^{[1]}$ are formed from $2\times2$ Haar unitary matrices, with different $2\times2$ blocks independent.
Then \eqref{eqn:lowerbound} gives
\begin{align*}
\E|\langle x|UU^T|y\rangle|^2&=\E\sum_{\alpha=1}^n|U_{x\alpha}|^2|U_{y\alpha}|^2\ge\frac{1-\varepsilon_n}{n},\numberthis\label{eqn:uutlower}
\end{align*}
which with \eqref{eqn:s2entropy} and \eqref{eqn:trw} implies \eqref{eqn:as2}.

For depth $<p=\log_2n$, we have not applied enough layers to act on all $p$ bits describing the modes. If we have not applied a layer for bit $k$, we can take the subsystem $\Gamma$ consisting of all modes whose $k$th bit is $0$, which is not connected in the circuit to any mode in $\Gamma^c$. Therefore $S_2(U)=0$ for this $\Gamma$.
\end{proof}

\section{Kaleidoscope circuit}\label{sec:kaleidoscope}

In this section, we prove Theorem~\ref{thm:kaleidoscope-s2}.
As for the butterfly circuit, Theorem~\ref{thm:kaleidoscope-s2} will follow if we prove the following convergence for the corresponding random walk $|X_t\rangle=U_t|X_0\rangle$, for $U_t$ a depth $t$ random kaleidoscope linear optical unitary, on the sphere.
The convergence proof for the kaleidoscope circuit ends up being more complicated due to the repeated layer structure in a single block (Figure~\ref{fig:kaleidoscope}), which complicates the coupling argument bounds. We do not end up with a clean contraction property such as Lemma~\ref{lem:contraction}, and instead will prove and use an iterative bound in Lemma~\ref{lem:cycle-iteration}.

\begin{thm}[kaleidoscope walk on the sphere]\label{thm:kaleidoscope-w2}
The kaleidoscope random walk $X_t$ mixes in $L^2$ Wasserstein distance in time $T=\Theta(\log n)$ on the sphere $\Sb^{n-1}$. More precisely, for any initial distribution $X_0\dsim \nu$ on the sphere, and for any $\kappa\in\N$ and $T\ge (160\kappa+40)\log_2n$,
\begin{align}
W_2(\mathcal L(X_T),\mu_n)&\le \frac{ 2^{1/4}(\log_2n)^{16\kappa+4}}{n^\kappa},
\end{align}
where $\mu_n$ denotes Haar measure on $\Sb^{n-1}$.
\end{thm}

Let $p=\log_2n$ and $t=t_0+p+\ell$ for some $t_0\in p\Z$ and $\ell\in\intbrr{0:p-1}$. 
Depending on whether $t_0=0$ or $p$ modulo $2p$, the sequence of $2p$ layers after $t_0$ is either $v=1,\ldots,p,p,\ldots,1$, where the value of $v$ indicates the layer connects modes separated by a bit flip in the $v$th position from the left, or $v=p,\ldots,1,1,\ldots,p$.
We will let $v(t)\in\{1,\ldots,p\}$ be such that layer $t$ connects modes separated by a bit flip in the $v(t)$th position from the left.

To prove the convergence in Theorem~\ref{thm:kaleidoscope-w2}, we consider the proportional coupling $(X_t,Y_t)$ as in \eqref{eqn:coupling4}, with the $2\times2$ unitaries in layer $t$ applied across modes $(i,i\oplus e_{v(t)})$, and let $\delta_t[i]:=|X_t[i]|^2-|Y_t[i]|^2$ as before. We would like to show that $\E\|\delta_t\|_2^2$ decreases exponentially in $t$ for large enough $t$. We will do this by comparing to a maximum over $\E\|\delta_{t'}\|_2^2$ for certain $t'<t$.
We will implicitly assume $p\ge2$, otherwise Theorem~\ref{thm:kaleidoscope-w2} is trivial when $p=1$ since $U_t$ is then always distributed as a $2\times 2$ Haar random unitary matrix.

Let the expectation notation $\E_{(i,j)}$ indicate that the most recent time step applies a random gate between modes $i$ and $j$. The main new calculation we need is that of ``duplicated'' connectivity; recalling that $\theta=\arcsin\sqrt{\zeta}$ with $\zeta\dsim\operatorname{Unif}[0,1)$, we evaluate
\begin{align*}
\E_{(i,j)}[\delta_t[i]\delta_t[j]|\mathcal F_{t-1}]
&=\E (\delta_{t-1}[i]+\delta_{t-1}[j])^2\cos^2\theta\sin^2\theta\\
&=\frac16(\delta_{t-1}[i]+\delta_{t-1}[j])^2.\numberthis\label{eqn:repeated}
\end{align*}
For the butterfly circuit in Section~\ref{sec:butterfly}, we did not encounter this calculation.

As in \eqref{eqn:iter1}, we have
\begin{align*}
\E[\|\delta_t\|_2^2]&=\frac23\E\|\delta_{t-1}\|_2^2+\frac23\sum_{i=0}^{2^p-1}\E\delta_{t-1}[i]\delta_{t-1}[i\oplus e_{v(t)}].\numberthis\label{eqn:initialrun}
\end{align*}
When $i\ne j$ and $i\oplus e_s\ne j$, we have, just as used for the butterfly circuit,
\begin{align}
\E_{(i,i\oplus e_s),(j,j\oplus e_s)}[\delta_t[i]\delta_t[j]|\mathcal F_{t-1}]
&=\sum_{r_s,r_s'\in\{0,1\}}\frac14\delta_{t-1}[i\oplus r_s]\delta_{t-1}[j\oplus r_s'].\numberthis\label{eqn:i2}
\end{align}
However, when iterating $\E\delta_{t-1}[i]\delta_{t-1}[i\oplus e_{v(t)}]$ through a full cycle of the layers, there will be duplicate connectivity when we reach the other $v=v(t)$ layer. 
We will have to use \eqref{eqn:repeated}, which shows that
\begin{align*}
\E_{(i,i\oplus e_s)}[\delta_t[i]\delta_t[i\oplus e_s]|\mathcal F_{t-1}]&=\frac16(\delta_{t-1}[i]+\delta_{t-1}[i\oplus e_s])^2\\
&\le\sum_{r_s,r_s'\in\{0,1\}}\frac14\delta_{t-1}[i\oplus r_s]\delta_{t-1}[i\oplus r_s']. \numberthis\label{eqn:iflip}
\end{align*}
Also, we will need to use (as in \eqref{eqn:split2}),
\begin{align*}
\E_{(i,i\oplus e_s)}[\delta_t[i]^2|\mathcal F_{t-1}]&=\frac43\sum_{r_s,r_s'\in\{0,1\}}\frac14\delta_{t-1}[i\oplus r_s]\delta_{t-1}[i\oplus r_s'].\numberthis\label{eqn:idouble}
\end{align*}

To estimate \eqref{eqn:initialrun}, we will want to iterate the second term through all $p$ types of layers, so that we can utilize the full connectivity of the network. The result of this iteration on the second term of \eqref{eqn:initialrun} is the following. 
\begin{lem}[cycle iteration]\label{lem:cycle-iteration}
Let $t\ge p$ and $[t]_p:=t\mod p\in\intbrr{0:p-1}$. Then
\begin{align}\label{eqn:kcycle}
\sum_{i=0}^{2^p-1}\E\delta_{t-1}[i]\delta_{t-1}[i\oplus e_{v(t)}]&\le
\sum_{j=1}^{p-[t]_p}\frac{\mathbf{1}_{[t]_p\ne0}}{2^{[t]_p+j+1}}\E\|\delta_{t-2[t]_p-j+1}\|_2^2.
\end{align}
\end{lem}
\begin{proof}
Let $\ell:=[t]_p$ for notational convenience.
We first note that if $\ell=0$ so $t\in p\Z$, then the same argument as for the butterfly circuit shows that $\sum_{i=0}^{2^p-1}\E\delta_{t-1}[i]\delta_{t-1}[i\oplus e_{v(t)}]\le0$, so \eqref{eqn:kcycle} holds. Otherwise, we do the following. 
Write $t=t_0+p+\ell$, where $t_0\in p\Z$. We will iterate down to $t_0$, which will cover all $p$ types of layers.
Without loss of generality, the order of layers from $t_0$ will be $v=p,\ldots,1,1,\ldots,p$, where the value of $v$ indicates the layer connects modes separated by a bit flip in the $v$th position from the left (Figure~\ref{fig:time}). Then for $t>t_0+p$, we have $v(t)=\ell=t\mod p$. The proof is the same in the other case where the sequence is $v=1,\ldots,p,p,\ldots,1$, just with different indexing.
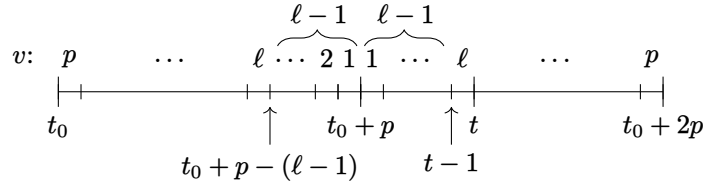
\begin{figure}[htb]
\begin{tikzpicture}
\draw(0,0)--(8,0);
\foreach \largetime/\tlab in {0/$t_0$,4/$t_0+p$,5.5/$t$,8/$t_0+2p$}{
    \draw(\largetime,.2)--++(0,-.4) node[below] {\tlab};
}
\def\twidth{.3} 
\foreach \smallt/\vlab in {0/$p$, {4-\twidth}/$1$, {4-2*\twidth}/$2$, 4/$1$,  {5.5-\twidth}/$\ell$, {2.5}/$\ell$, {8-\twidth}/$p$}{
    \draw(\smallt+\twidth,-.1)--++(0,.2);
    \draw(\smallt,-.1)--++(0,.2);
    \node[above] at (\smallt+\twidth/2,.2) {\vlab};
}
\foreach \dotloc in {1.5,3.1,4.75,6.6}{ 
    \node[above] at (\dotloc,.2) {$\cdots$};
}
\node[above left] at (-.2,.25) {$v$:};
\draw[<-] (2.5+\twidth,-.2)--++(0,-.5) node[below] {$t_0+p-(\ell-1)$};
\draw[<-] (5.5-\twidth,-.2)--++(0,-.5) node[below] {$t-1$};
\foreach \lbrace in {0,1.15}{
    \draw[xshift=\lbrace cm,decoration={brace,raise=3pt,aspect=.5,amplitude=6pt},decorate] (2.6+\twidth,.5)--(4-.05,.5);
    \node[xshift=\lbrace cm,above] at (3.45,.75) {$\ell-1$};
}
\end{tikzpicture}
\caption{Relationship between time index and layers of the kaleidoscope circuit, with some important times used in Lemma~\ref{lem:cycle-iteration} marked.}\label{fig:time}
\end{figure}
\begin{enumerate}[leftmargin=*]
\item First iterate $\E\delta_{t-1}[i]\delta_{t-1}[i\oplus e_{\ell}]$ down $\ell-1$ times to reach the midpoint time $t_0+p$. This is just like for the butterfly circuit; there are no layer or bit string overlaps since we never apply a layer connecting across the $\ell$th bit, and we generate a $2^{-2}$ factor each time. This iteration is described by \eqref{eqn:i2}. We obtain
\begin{align*}
\sum_{i=0}^{2^p-1}&\E\delta_{t-1}[i]\delta_{t-1}[i\oplus e_\ell]\\
&=\frac1{2^{2(\ell-1)}}\sum_{i=0}^{2^p-1}\sum_{r_1,\ldots,r_{\ell-1},r_1',\ldots,r_{\ell-1}'\in\{0,1\}}\E\delta_{t_0+p}[i\oplus r_{\ell-1}\oplus \cdots\oplus r_1]\delta_{t_0+p}[i\oplus e_{\ell}\oplus r_{\ell-1}'\oplus\cdots\oplus r_1']\\
&=\frac1{2^{\ell-1}}\sum_{i=0}^{2^p-1}\sum_{r_1',\ldots,r_{\ell-1}'\in\{0,1\}}\E\delta_{t_0+p}[i]\delta_{t_0+p}[i\oplus e_{\ell}\oplus r_{\ell-1}'\oplus\cdots\oplus r_1'].\numberthis
\end{align*}

\item 
\begin{enumerate}[(a)]
\item We can iterate $\ell-1$ more times before reaching the layer $t_0+p-(\ell-1)$, which connects modes that differ by a bit in the $\ell$th position (Figure~\ref{fig:time}). The layer structures for these $\ell-1$ iterations before that have already been applied, but since the arguments in the $\delta_{t_0+p-j}$ functions always differ in the $\ell$th bit, we never have to invoke \eqref{eqn:iflip} or \eqref{eqn:idouble}. We can just use the usual \eqref{eqn:i2}, although we see no $2^{-2}$ factors are generated in total, since we can sum out duplicate $r_j$'s and $r_j'$'s. This gives that the previous equation is
\begin{align*}
&=\frac1{2^{\ell-1}}\sum_{i=0}^{2^p-1}\sum_{r_1',\ldots,r_{\ell-1}'\in\{0,1\}}\E\delta_{t_0+p-(\ell-1)}[i]\delta_{t_0+p-(\ell-1)}[i\oplus e_{\ell}\oplus r_{\ell-1}'\oplus\cdots\oplus r_1'].\numberthis
\end{align*}

\item Now we iterate through layer $t_0+p-(\ell-1)$, which connects modes differing by the $\ell$th bit. We have to use the inequality \eqref{eqn:iflip} when all $r_j'=0$ in the above, but we do not need to use \eqref{eqn:idouble} yet. We iterate once to get that the above expression is
\begin{align*}
&\le\begin{multlined}[t]
\frac1{2^{\ell-1}}\sum_{i=0}^{2^p-1}\sum_{r_1',\ldots,r_{\ell-1}',r_{\ell}',r_{\ell}\in\{0,1\}}\frac14\E\delta_{t_0+p-\ell}[i\oplus r_{\ell}]\delta_{t_0+p-\ell}[i\oplus e_{\ell}\oplus r_{\ell-1}'\oplus\cdots\oplus r_1'\oplus r_\ell']
\end{multlined}\\
&=\begin{multlined}[t]
\frac1{2^{\ell}}\sum_{i=0}^{2^p-1}\sum_{r_1',\ldots,r_{\ell-1}',r_{\ell}'\in\{0,1\}}\E\delta_{t_0+p-\ell}[i]\delta_{t_0+p-\ell}[i\oplus r_{\ell-1}'\oplus\cdots\oplus r_1'\oplus r_{\ell}'].
\end{multlined}\numberthis
\end{align*}

\end{enumerate}

\item There are $p-\ell$ remaining layers to go until reaching $t_0$. There are no connectivity overlaps because the gates are new connections for bit positions $p,p-1,\ldots,\ell+1$. So there is no use of \eqref{eqn:iflip}, but instead we have to use \eqref{eqn:idouble} when all $r_s'=0$. This generates an extra $\frac14\E[\delta_{t-k}[i]^2]$ from each iteration. We see the above becomes
\begin{multline}\label{eqn:iter0}
\frac1{2^{p}}\sum_{i=0}^{2^p-1}\sum_{r_1',\ldots,r_{\ell-1}',r_{\ell}',\ldots,r_p'\in\{0,1\}}\E\delta_{t_0}[i]\delta_{t_0}[i\oplus r_{\ell-1}'\oplus\cdots\oplus r_1'\oplus r_{\ell}'\oplus\cdots\oplus r_p']+\\
+\sum_{j=1}^{p-\ell}\sum_{i=0}^{2^p-1}\frac1{2^{\ell+j+1}}\E\delta_{t_0+p-\ell-(j-1)}[i]^2.
\end{multline}
The first term in the above is
\begin{align*}
\frac1{2^p}\sum_{i=0}^{2^p-1}\sum_{r_1',\ldots,r_{\ell-1}',r_{\ell}',\ldots,r_p'\in\{0,1\}}\E\delta_{t_0}[i]\delta_{t_0}[i\oplus r_{\ell-1}'\oplus\cdots\oplus r_1'\oplus r_{\ell}'\oplus\cdots\oplus r_p']&=\frac1{2^p}\sum_{i=0}^{2^p-1}\sum_{j=0}^{2^p-1}\E\delta_{t_0}[i]\delta_{t_0}[j]\\
&=\frac1{2^p}\E\left(\sum_{i=0}^{2^p-1}\delta_{t_0}[i]\right)^2=0.
\end{align*}
So the remaining term in \eqref{eqn:iter0} is simply
\begin{align*}
\sum_{j=1}^{p-\ell}\sum_{i=0}^{2^p-1}\frac1{2^{\ell+j+1}}\E\delta_{t_0+p-\ell+1-j}[i]^2&=\sum_{j=1}^{p-\ell}\frac1{2^{\ell+j+1}}\E\|\delta_{t-2\ell-j+1}\|_2^2.
\end{align*}
\end{enumerate}
\end{proof}

\begin{proof}[Proof of Theorem~\ref{thm:kaleidoscope-w2}]
Let $p=\log_2n$, and set $[t]_p:=t\mod p\in\intbrr{0:p-1}$. 
Starting from \eqref{eqn:initialrun} and applying Lemma~\ref{lem:cycle-iteration}, we see for any $t\ge p$,
\begin{align*}
\E[\|\delta_{t}\|_2^2]&\le\frac23\E\|\delta_{t-1}\|_2^2+\frac23\sum_{j=1}^{p-[t]_p}\frac{\mathbf{1}_{[t]_p\ne0}}{2^{[t]_p+j+1}}\E\|\delta_{t-2[t]_p-(j-1)}\|_2^2\\
&\le \frac56\max\bigg(\E\|\delta_{t-1}\|_2^2,\sum_{j=1}^{p-[t]_p}\frac{\mathbf{1}_{[t]_p\ne0}}{2^{[t]_p+(j-1)}}\E\|\delta_{t-2[t]_p-(j-1)}\|_2^2\bigg).\numberthis\label{eqn:iter-max}
\end{align*}
Now we iterate down in time. We want to know how far back in time we have to go to obtain an exponential prefactor of the form e.g. $\left(\frac56\right)^{p/2}$. 
Let $t\ge 5p$. We will show that there is $s(t)\le 5p$ such that
\begin{align}\label{eqn:st-goal}
\E\|\delta_t\|_2^2&\le \left(\frac56\right)^{p/2}p^2\max_{k\in\intbrr{0:s(t)}}\E\|\delta_{t-s(t)+k}\|_2^2.
\end{align}

Starting from \eqref{eqn:iter-max}, first let $s'\ge0$ be how many times at the start that the maximum is achieved by the first term (non-sum term) in repeated applications of \eqref{eqn:iter-max}.
Then
\begin{align}
\E\|\delta_{t}\|_2^2&\le\left(\frac56\right)^{s'}\E\|\delta_{t-s'}\|_2^2.
\end{align}
If $s'\ge p/2$, then we have obtained the desired prefactor and \eqref{eqn:st-goal} holds with $s(t)=\lceil p/2\rceil$. 
Otherwise, we continue, knowing the next iteration involves the second term (sum term) in \eqref{eqn:iter-max}, which also means we have $\ell':=(t-s')\mod p\ne0$. 
We perform this iteration to obtain 
\begin{align}
\E\|\delta_t\|_2^2&\le\left(\frac56\right)^{s'+1}\sum_{j=1}^{p-\ell'}\frac1{2^{\ell'+(j-1)}}\E\|\delta_{t-s'-2\ell'-(j-1)}\|_2^2. 
\end{align}
We continue iterating each term $\E\|\delta_{t-s'-2\ell'-(j-1)}\|_2^2$ using \eqref{eqn:iter-max}, and will let $s_j''\ge0$ be the corresponding number of next consecutive times that the maximum is attained by the first (non-sum) term in the applications of \eqref{eqn:iter-max}. This gives
\begin{align}\label{eqn:sdprime}
\E\|\delta_{t}\|_2^2&\le \left(\frac56\right)^{s'+1}\sum_{j=1}^{p-\ell'}\frac1{2^{\ell'+(j-1)}}\left(\frac56\right)^{s_j''}\E\|\delta_{t-s'-s_j''-2\ell'-(j-1)}\|_2^2.
\end{align}
For $j$ with $s'+s_j''\ge \lceil p/2\rceil+2$, we will stop\footnote{Up to acquiring additional constant factors, we could stop when $s'+s_j''=\lceil p/2\rceil-1$, but we require slightly larger $s_j''$ here to make some later expressions a bit simpler or more convenient.}
 $s_j''$ when $s'+s_j''=\lceil p/2\rceil +2$.
For such terms in \eqref{eqn:sdprime}, the smallest index in $\delta_j$ is then $t-(\lceil p/2\rceil+2)-2\ell'-(p-\ell'-1)\ge t-(3p+1)$, which will be consistent with $s(t)\le 5p$ in \eqref{eqn:st-goal}. Also, note that $\sum_{j=1}^{p-\ell'}\frac1{2^{\ell'+(j-1)}}\le2$.

For the remaining $j$, we iterate one more time.
The next iteration involves the second term in \eqref{eqn:iter-max} with $\ell_{j}=(t-s'-s_j''-2\ell'-(j-1))\mod p=(-s_j''-\ell'-(j-1))\mod p\in\intbrr{0:p-1}$ which must be nonzero, giving
\begin{multline}
\E\|\delta_{t}\|_2^2\le\left(\frac56\right)^{s'+2}\sum_{j=1}^{p-\ell'}\frac{\mathbf{1}_{s'+s_j''<\lceil p/2\rceil +2}}{2^{\ell'+(j-1)}}\left(\frac56\right)^{s_j''}\sum_{j'=1}^{p-\ell_j}\frac{1}{2^{\ell_j+(j'-1)}}\E\|\delta_{t-s'-s_j''-2\ell'-(j-1)-2\ell_j-(j'-1)}\|_2^2\\
+2\left(\frac56\right)^{p/2+3}\max_{j\in\intbrr{1:p}}\E\|\delta_{t-\lceil p/2\rceil-2\ell'-j-1}\|_2^2.
\end{multline}
We count the exponent powers in the first term. Recall that $\ell'=(t-s')\mod p$ and is nonzero. We claim that for any $j\in\intbrr{1:p-\ell'}$,
\begin{align}\label{eqn:claim-exp}
s_j''+\ell'+\ell_j+(j-1)\ge p/2.
\end{align}
Indeed, suppose $1\le s_j''+\ell'+(j-1)<p/2$. Then $\ell_j=(p-s_j''-\ell'-(j-1))\mod p\ge p/2$.
Then using \eqref{eqn:claim-exp} and that $1/2<5/6$, we get
\begin{align*}
\E\|\delta_t\|_2^2&\le
\begin{multlined}[t]
\left(\frac56\right)^{s'+2+p/2}p^2\max_{\substack{j\in\intbrr{1:p-\ell'},j'\in\intbrr{1:p-\ell_j}:\\s'+s_j''<\lceil p/2\rceil+2}}(\E\|\delta_{t-s'-s_j''-2\ell'-(j-1)-2\ell_j-(j'-1)}\|_2^2)\\
+2\left(\frac56\right)^{p/2+3}\max_{j\in\intbrr{1:p}}\E\|\delta_{t-\lceil p/2\rceil-2\ell'-j-1}\|_2^2
\end{multlined}\\
&\le \left(\frac56\right)^{p/2}p^2\max_{k\in\intbrr{0:\lceil p/2\rceil+4p}}\E\|\delta_{t-\lceil p/2\rceil-4p+k}\|_2^2,\numberthis
\end{align*}
using that $(5/6)^2p^2+2(5/6)^3<p^2$ for $p\ge2$. We can then take $s(t)\le 5p$ in \eqref{eqn:st-goal}.

Equation~\eqref{eqn:st-goal} shows we can obtain an exponential prefactor $\left(\frac56\right)^{p/2}$ by iterating backwards in time at most $5p=5\log_2n$ steps. Continuing to apply \eqref{eqn:st-goal} to each term in the maximum, and using the bound $\E\|\delta_s\|_2^2\le 2$ for any $s$, we then obtain for $t\ge 5C_1p$, with $C_1\in\N$, 
\begin{align}
\E\|\delta_t\|_2^2&\le 2\left(\frac56\right)^{C_1p/2}p^{2C_1}.
\end{align}
Similarly as in \eqref{eqn:w2bound}, we then have
\begin{align*}
W_2(\mathcal L(X_t),\mu_n)&\le (n\E\|\delta_t\|_2^2)^{1/4}
\le (2n)^{1/4}\left(\frac56\right)^{C_1p/8}p^{C_1/2}.\numberthis\label{eqn:w2bound-kal}
\end{align*}
For $C_1=32\kappa+8$, \eqref{eqn:w2bound-kal} is $\le 2^{1/4}(\log_2n)^{16\kappa+4}/n^\kappa$.
\end{proof}

\begin{proof}[Proof of Theorem~\ref{thm:kaleidoscope-s2}]
This follows the proof of Theorem~\ref{thm:butterfly-s2}. The only change is as follows.
Theorem~\ref{thm:kaleidoscope-w2} implies that for $t\ge (160\kappa+40)\log_2n$, $X_t$  is $2^{1/4}(\log_2n)^{16\kappa+4}n^{-\kappa}$-close in $L^2$ Wasserstein distance to Haar random on $\Sb^{n-1}$. Take $\kappa=3$ so that $160\kappa+40=520$ and $16\kappa+4=52$. Then similarly to \eqref{eqn:lowerbound}, we obtain for depth $t\ge 520\log_2n$,
\begin{align}
\E[|U_{x\alpha}|^2|U_{y\alpha}|^2]&\ge\frac{1}{n(n+1)}-\frac{4(2^{1/4})(\log_2n)^{52}}{n^{3}}
\ge \frac{1-\varepsilon_n}{n^2},
\end{align}
for $\varepsilon_n=n^{-1}+4(2^{1/4})(\log_2n)^{52}n^{-1}$.
The rest of the proof is the same.
\end{proof}

\section{Typicality}\label{sec:typicality}

The entanglement bounds in Theorems~\ref{thm:butterfly-s2} and \ref{thm:kaleidoscope-s2} are for average subsystem entanglement in any subsystem. We can obtain some brief typicality results for a subsystem $\Gamma$.
\begin{prop}[typicality]\label{prop:typicality}
Let $\Gamma\subset\{0,\ldots,2^p-1\}$, and consider the random butterfly or kaleidoscope circuit at a depth $T\ge C\log_2n$ as assumed in the appropriate Theorem~\ref{thm:butterfly-s2} or \ref{thm:kaleidoscope-s2}. Let $M_\Gamma=\min(|\Gamma|,n-|\Gamma|)\log\cosh(2s)$, which is the maximum R\'enyi-2 entropy for $\Gamma$. Then if $|\Gamma|=o(n)$, there is a sequence $\delta_n=\delta_n(|\Gamma|)=o(1)$ such that
\begin{align}\label{eqn:ty-on}
\P[S_2(U)\ge (1-\delta_n)M_\Gamma] &\ge 1-o(1).
\end{align}
More generally, let $r:=|\Gamma|/n$ and $0<\alpha<1$. Then for any $|\Gamma|\le n/2$,
\begin{align}\label{eqn:ty-alpha}
\P[S_2(U)\ge \alpha M_\Gamma] &\ge 1-\frac{r(1-\varepsilon_n)+\varepsilon_n}{1-\alpha},
\end{align}
for $\varepsilon_n=o(1)$ as in the appropriate Theorem~\ref{thm:butterfly-s2} or \ref{thm:kaleidoscope-s2}.
\end{prop}
The bound \eqref{eqn:ty-on} can be understood as follows. For any subsystem $\Gamma$ of size $|\Gamma|=o(n)$, the average bound \eqref{eqn:as2} or \eqref{eqn:as2-kaleidoscope} for $\E S_2(U)$ is $|\Gamma|\log(\cosh(2s))(1-o(1))$, which is only a factor $1-o(1)$ away from the maximum possible R\'enyi-2 entropy $M_\Gamma$ for $\Gamma$. Therefore typical instances of $S_2(U)$ must also be very close to $M_\Gamma$ with high probability.
The bound \eqref{eqn:ty-alpha} for extensive $\Gamma$ is a weak bound, and in general is just bounded away from zero (as long as one chooses appropriate $\alpha$).

\begin{proof}[Proof of Proposition~\ref{prop:typicality}]
Write
\begin{align*}
\E S_2(U)=\int_0^{M_\Gamma}\P[S_2(U)>t]\,dt
&\le \int_0^{\alpha M_\Gamma}1\,dt +\int_{\alpha M_\Gamma}^{M_\Gamma}\P[S_2(U)>t]\,dt\\
&\le \alpha M_\Gamma+(1-\alpha)M_\Gamma \P[S_2(U)>\alpha M_\Gamma].\numberthis
\end{align*}
Then applying \eqref{eqn:as2} or \eqref{eqn:as2-kaleidoscope} as appropriate, we obtain
\begin{align*}
\P[S_2(U)\ge\alpha M_\Gamma]&\ge \frac{(1-r)(1-\varepsilon_n)-\alpha }{1-\alpha}=1-\frac{r(1-\varepsilon_n)+\varepsilon_n}{1-\alpha}.
\end{align*}
If $|\Gamma|=o(n)$, so that $r=r_n=|\Gamma|/n=o(1)$, then take e.g. $\alpha=\alpha_n=1-\max(\sqrt{r_n},\sqrt{\varepsilon_n})$ to obtain \eqref{eqn:ty-on}.
\end{proof}

\vspace{2mm}
\noindent
\textbf{Acknowledgments.}
We thank Victor Albert, Victor Galitski, Joseph Iosue, Amit Vikram, and Yu-Xin Wang for discussions.
L.S.~and A.V.G.~ acknowledge support from the U.S.~Department of Energy, Office of Science, Accelerated Research in Quantum Computing, Fundamental Algorithmic Research toward Quantum Utility (FAR-Qu).
L.S.~and~A.V.G.~were also supported in part by ARL (W911NF-24-2-0107), ONR MURI, NSF QLCI (award No.~OMA-2120757), DoE ASCR Quantum Testbed Pathfinder program (award No.~DE-SC0024220), NSF STAQ program, AFOSR MURI,  and NQVL:QSTD:Design:FTL. L.S.~and A.V.G.~also acknowledge support from the U.S.~Department of Energy, Office of Science, National Quantum Information Science Research Centers, Quantum Systems Accelerator (award No.~DE-SCL0000121).

\bibliographystyle{amsalpha_edit}
\bibliography{entanglement.bib}
\end{document}